\documentclass{llncs}
\usepackage{amsmath}
\usepackage[utf8]{inputenc}
\usepackage{amstext}
\usepackage{amssymb}
\usepackage[langfont=caps]{complexity}
\usepackage{authblk}
\usepackage{tikz}

\newcommand{\s}{\sum}

\renewcommand{\p}{\prod}

\newcommand{\pn}[1]{\prod^{[#1]}}

\begin{document}

\title{Improved bounds for reduction to depth $4$ and depth $3$}
\date{\today}
\author{Sébastien Tavenas}
\institute{LIP\thanks{UMR 5668 ENS Lyon - CNRS - UCBL - INRIA,
    Université de Lyon, sebastien.tavenas@ens-lyon.fr}, École Normale Supérieure de Lyon}

\maketitle

\begin{abstract}
  Koiran~\cite{Koi12} showed that if an $n$-variate polynomial $f_n$
  of degree $d$ (with $d=n^{O(1)}$) is computed by a circuit of size
  $s$, then it is also computed by a homogeneous circuit of depth four
  and of size $2^{O(\sqrt{d}\log(n)\log(s))}$. Using this result,
  Gupta, Kamath, Kayal and Saptharishi~\cite{GKKS13} found an upper
  bound for the size of a depth three circuit computing $f_n$.

  We improve here Koiran’s bound. Indeed, we show that it is possible
  to transform an arithmetic circuit into a depth four circuit of size
  $2^{\left(O\left(\sqrt{d\log (ds)\log (n)}\right)\right)}$. Then,
  mimicking the proof in~\cite{GKKS13}, it also implies an
  $2^{\left(O\left(\sqrt{d\log (ds)\log (n)}\right)\right)}$ upper
  bound for depth three circuits.

  This new bound is not far from optimal in the sense that Gupta,
  Kamath, Kayal and Saptharishi~\cite{GKKS12} also showed a
  $2^{\Omega(\sqrt{d})}$ lower bound for the size of homogeneous depth
  four circuits such that gates at the bottom have fan-in at most
  $\sqrt{d}$. Finally, we show that this last lower bound also holds
  if the fan-in is at least $\sqrt{d}$.
\end{abstract}

\section{Introduction}

Valiant, Skyum, Berkowitz and Rackoff~\cite{VSBR83} proved that if a
size-$s$ depth-$d$ circuit computes a polynomial of degree $d$, then
this polynomial can also be computed by a circuit of depth
$O(\log(d)\log(s))$ and of size bounded by a polynomial in $s$ (this
result will be the basis for the parallelization in this paper). Some
years later, Allender, Jiao, Mahajan and Vinay~\cite{AJMV98} studied
this parallelization method and showed it could be done
uniformly. Using the proof of these results, Agrawal and Vinay proved~\cite{AV08}
that if an n-variate polynomial $f$ of degree $d=O(n)$ has a circuit
of size $2^{o(d+d\log(\frac nd))}$, then $f$ can also be computed by a
depth-four circuit ($\sum\prod\sum\prod$) of size $2^{o(d+d\log(\frac
  nd))}$. This result shows that for proving arithmetic circuit lower
bounds or black-box derandomization of identity testing, the case of
depth four arithmetic circuit is the general case in a certain sense.
 
The hypothesis of Agrawal and Vinay’s result is quite weak: they
consider circuits of size $2^{o(d+d \log( \frac{n}{d} ))}$ (we can
notice that all polynomials have a formula of size $d\binom{n+d}{d}=2^
{O(d\log(\frac{n+d}{d}))}$). But if the hypothesis is strengthened, it
is possible to get a stronger conclusion. Indeed, Koiran~\cite{Koi12}
showed that if a circuit is of size $s$, then it can be computed by a
homogeneous depth-four circuit of size
$2^{O(\sqrt{d}\log(d)\log(s))}$. For example, if the permanent family
is computed by a polynomial size circuit (i.e., of size $n^c$), then
it is computed by a depth-four circuit of size
$2^{O(\sqrt{n}\log^2(n))}$. These results appear as an interesting
approach to lower bounds: if one finds a
$2^{\omega\left(\sqrt{n}\log^2(n)\right)}$ lower bound on the size of
$\s\pn{O\left(\sqrt{n}\right)}\s\pn{\sqrt{n}}$ circuits computing the
permanent, then it will imply that there are no polynomial size
circuits for the permanent. Moreover it could be easier to find lower
bounds on the size of these particular circuits than for the general
circuits. Indeed, although no superpolynomial lower bound is known for
general circuits, Gupta, Kamath, Kayal and Saptharishi~\cite{GKKS12}
get a nearly optimal lower bound for particular depth-4 circuits for
the permanent.  More precisely, they showed that if a homogeneous
$\sum \prod \sum \prod$ circuit where the bottom fan-in is bounded by
$t$ computes the permanent of a matrix of size $n\times n$, then its
size is $2^{\Omega(\frac n t)}$. In particular, a
$\s\pn{O\left(\sqrt{n}\right)}\s\pn{\sqrt{n}}$ circuit computing the
permanent is of size $2^ {\Omega\left(\sqrt{n}\right)}$. The following
year, the same authors~\cite{GKKS13} improve the upper bound by
transforming $n$-variate circuits of size $s$ and depth $d$
($=n^{O(1)}$) into depth-$3$ circuits of size $2^{\left(O(\sqrt{d\log
      s\log n\log d})\right)}$, moreover if the input is a branching
program (and not a circuit), the upper bound becomes
$2^{\left(O(\sqrt{d\log s\log n})\right)}$. In particular, this result
gives a depth-$3$ circuit of size $2^{O\left(\sqrt{n}\log n\right)}$
computing the determinant of a matrix $n\times n$. Nevertheless, this
result is not comparable to the depth-$4$ reductions since the
depth-$3$ circuit they get is not homogeneous, and uses gates
computing polynomials of very high degree. Very recently, Fournier,
Limaye, Malod and Srinivasan~\cite{FLMS13} showed an
$2^{\Omega(\sqrt{d/t}\log n)}$-lower bound for the size of the
$\s\p\s\p$ circuits, with bottom fan-in bounded by $t$, which compute
the iterated matrix multiplication.

In this paper we improve Koiran's bound. We show that a circuit of
size $s$ can be parallelized homogeneously in depth $4$ and in size
$2^{\left(O\left(\sqrt{d\log(ds)\log(n)}\right)\right)}$ such that the
fan-in of each multiplication gate is bounded by
$O\left(\sqrt{d\frac{\log ds}{\log n}}\right)$. We can notice that as
$n\leq s$, the result implies Koiran's bound and is generally better
(in the case where $d,s=n^{\Theta(1)}$, Koiran's bound is
$2^{O(\sqrt{n}\log^2 n)}$ while the new bound is $2^{O(\sqrt{n}\log
  n)}$). It implies that a $2^{\omega\left(\sqrt{n}\log(n)\right)}$
lower bound for depth-$4$ circuits computing the permanent gives a
super-polynomial lower bound for general circuits computing the
permanent. Moreover, using this result in Gupta, Kamath, Kayal and
Saptharishi's proof instead of Koiran's result slightly improves the
depth-$3$ upper bound. An $n$-variate circuit of size $s$ and depth
$d$ is computed by a depth-$3$ circuit of size
$2^{\left(O(\sqrt{d\log(ds)\log n})\right)}$. So, we get the same
bound for the reduction at depth $3$ starting from an arithmetic
circuit as from an arithmetic branching program. Finally in
Section~\ref{Sec_lowerbound}, we show, by a counting argument, that if
a homogeneous $\sum \prod \sum \prod$ circuit where the bottom fan-in
is lower-bounded by $t$ computes the permanent (or the determinant) of
a matrix of size $n\times n$, then its size is $2^{\Omega(t\log n)}$.

\section{Arithmetic Circuits}

We give here a brief introduction to the theory of arithmetic circuits. The
reader can find more detailed information
in~\cite{Gat87,Bur00,SY10,CKW11}. In this theory, we measure the
complexity of polynomial functions using arithmetic circuits.

\begin{definition}
  An arithmetic circuit is a finite acyclic directed graph with
  vertices of in-degree $0$ or more and exactly one vertex of
  out-degree $0$. Vertices of in-degree $0$ are called inputs and
  labeled by a constant or a variable. The other vertices are labeled
  by $\times$ or $+$ (or sometimes by $\odot$ in this paper) and
  called computation gates (the in-degree of these gates will be also
  called the fan-in). The vertex of out-degree $0$ is called the
  output. The vertices of a circuit are commonly called gates and its
  edges arrows. Finally, we call a formula, an arithmetic circuit such
  that the underlying graph is a tree.
\end{definition}

A $\odot$-gate corresponds to a multiplication-by-a-scalar
gate. The fan-in of such a gate will be always $2$ and at least one of
its inputs corresponds to a constant (we will give a syntactic
restriction just after the next definition).

Each gate of a circuit computes a polynomial (defined by
induction). The polynomial computed by a circuit corresponds to the
polynomial computed by the output of this circuit. For a gate
$\alpha$, we denote $[\alpha]$ the polynomial computed by this gate. In fact, for some proofs, we will use
circuits with several outputs (each one corresponds to an out-degree
$0$ gate).  

\begin{definition}
  The size of a circuit is its number of gates. The depth is the
  maximal length of a directed path from an input to an output. The
  degree of a gate is defined recursively: constant inputs labelled by
  $0$ are of degree $-\infty$, other constant inputs are of degree
  $0$, any variable input is of degree 1, the degree of a $+$-gate is
  the maximum of the incoming degrees and the degree of a
  $\times$-gate (or a $\odot$-gate) is the sum of the incoming
  degrees.

  A circuit is called {\it homogeneous} is for each $+$-gate $\alpha$,
  all the inputs of $\alpha$ have same degree.
\end{definition}

We can now put a restriction for the $\odot$-gates. For each one of
these gates, one of its two children has to be of degree $0$.

\begin{remark}\label{rem_zero}
  In the following, we will assume that the computation gates will
  never compute the zero polynomial. If it is the case, it is
  sufficient to replace this gate by an input gate labelled by the
  constant $0$.
\end{remark}

For a given circuit we will consider graphs called parse trees. A
parse tree corresponds, in the spirit, to the computation of one
particular monomial.

\begin{definition}
  The set of parse trees of a circuit $C$ is defined by induction on
  its size:
  \begin{itemize}
  \item If $C$ is of size $1$ it has only one parse tree, itself.
  \item If the output gate $o$ of $C$ is a $+$-gate whose inputs are
    the gates $\alpha_1,\ldots,\alpha_k$, then the parse trees of $C$
    are obtained by choosing, for an arbitrary $i\leq k$, a parse tree
    of the sub-circuit rooted in $\alpha_i$ and the arrow from
    $\alpha_i$ to the output $o$.
  \item If the output gate $o$ of $C$ is a $\times$-gate or an
    $\odot$-gate whose inputs are the gates
    $\alpha_1,\ldots,\alpha_k$, the parse trees of $C$ are obtained by
    taking for each $i\leq k$, one disjoint copie of a parse tree of
    the sub-circuit rooted in $\alpha_i$, and the arrows from all
    $\alpha_i$ to the output $o$.
  \end{itemize}
\end{definition}

For example, the following circuit
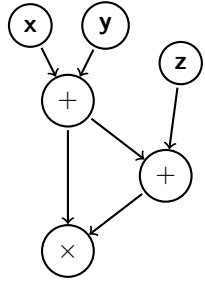
\begin{center}
  \tikzstyle{gat}=[thick,draw,circle,font=\sffamily\bfseries]
  \begin{tikzpicture}[shorten >=1pt]
    \node[gat] (0) at (0,0) {x};
    \node[gat] (1) at (1,0) {y};
    \node[gat] (2) at (2,-0.5) {z};
    \node[gat] (3) at (0.5,-1) {$+$}
    edge[thick,<-] (0)
    edge[thick,<-] (1);
    \node[gat] (4) at (1.8,-2) {$+$}
    edge[thick,<-] (2)
    edge[thick,<-] (3);
    \node[gat] (5) at (0.5,-3) {$\times$}
    edge[thick,<-] (3)
    edge[thick,<-] (4);
  \end{tikzpicture}
\end{center}
has six parse trees.
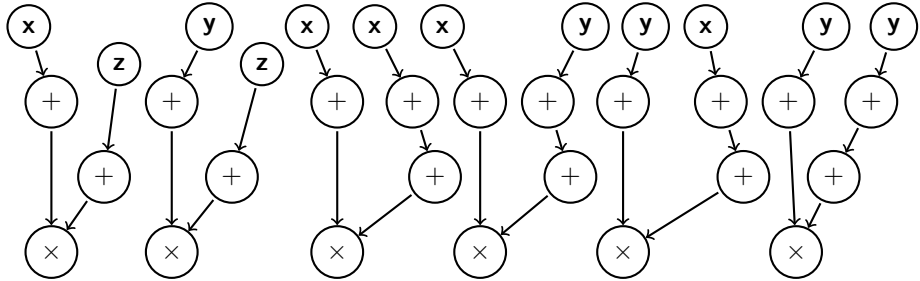
\begin{center}
  \tikzstyle{gat}=[thick,draw,circle,font=\sffamily\bfseries]
  \begin{tikzpicture}[shorten >=1pt]
    \node[gat] (0) at (0.2,0) {x};
    \node[gat] (2) at (1.4,-0.5) {z};
    \node[gat] (3) at (0.5,-1) {$+$}
    edge[thick,<-] (0);
    \node[gat] (4) at (1.2,-2) {$+$}
    edge[thick,<-] (2);
    \node[gat] (5) at (0.5,-3) {$\times$}
    edge[thick,<-] (3)
    edge[thick,<-] (4);

    \node[gat] (11) at (2.6,0) {y};
    \node[gat] (12) at (3.3,-0.5) {z};
    \node[gat] (13) at (2.1,-1) {$+$}
    edge[thick,<-] (11);
    \node[gat] (14) at (2.9,-2) {$+$}
    edge[thick,<-] (12);
    \node[gat] (15) at (2.1,-3) {$\times$}
    edge[thick,<-] (13)
    edge[thick,<-] (14);

    \node[gat] (20) at (3.9,0) {x};
    \node[gat] (21) at (4.8,0) {x};
    \node[gat] (23) at (4.3,-1) {$+$}
    edge[thick,<-] (20);
    \node[gat] (24) at (5.3,-1) {$+$}
    edge[thick,<-] (21);
    \node[gat] (26) at (5.6,-2) {$+$}
    edge[thick,<-] (24);
    \node[gat] (25) at (4.3,-3) {$\times$}
    edge[thick,<-] (23)
    edge[thick,<-] (26);

    \node[gat] (30) at (5.7,0) {x};
    \node[gat] (31) at (7.6,0) {y};
    \node[gat] (32) at (6.2,-1) {$+$}
    edge[thick,<-] (30);
    \node[gat] (33) at (7.1,-1) {$+$}
    edge[thick,<-] (31);
    \node[gat] (34) at (7.4,-2) {$+$}
    edge[thick,<-] (33);
    \node[gat] (35) at (6.2,-3) {$\times$}
    edge[thick,<-] (32)
    edge[thick,<-] (34);

    \node[gat] (40) at (8.4,0) {y};
    \node[gat] (41) at (9.2,0) {x};
    \node[gat] (42) at (8.1,-1) {$+$}
    edge[thick,<-] (40);
    \node[gat] (43) at (9.4,-1) {$+$}
    edge[thick,<-] (41);
    \node[gat] (44) at (9.7,-2) {$+$}
    edge[thick,<-] (43);
    \node[gat] (45) at (8.1,-3) {$\times$}
    edge[thick,<-] (42)
    edge[thick,<-] (44);

    \node[gat] (50) at (10.8,0) {y};
    \node[gat] (51) at (11.7,0) {y};
    \node[gat] (52) at (10.3,-1) {$+$}
    edge[thick,<-] (50);
    \node[gat] (53) at (11.4,-1) {$+$}
    edge[thick,<-] (51);
    \node[gat] (54) at (10.9,-2) {$+$}
    edge[thick,<-] (53);
    \node[gat] (55) at (10.4,-3) {$\times$}
    edge[thick,<-] (52)
    edge[thick,<-] (54);
  \end{tikzpicture}
\end{center}

We can notice that the size of a parse tree can be exponentially
larger that the one of the original circuit. It will not be a problem
in this paper. However, it is possible to avoid this increase  
using multiplicatively disjoint circuits as it is done in~\cite{MP08}.
    
At each parse tree, we can associate the monomial which corresponds to the
product of the leaves.

The next lemma is proved in~\cite{MP08}.

\begin{lemma}\label{lem_1.11} 
  A polynomial $f$ computed by a circuit $C$ equals
  the sum of the monomials of the parse trees:
  \begin{align*}
    f=\s_{\substack{T\textrm{ parse} \\ \textrm{tree}}} m(T)
  \end{align*}
  where $m(T)$ is the monomial associated to the tree $T$.
\end{lemma}

We will use some convenient notations which are defined in~\cite{GKKS13}. A
depth-$4$ circuit such that gates are multiplication gates at level one and
three and addition gates at levels two and four are denoted
$\sum\prod\sum\prod$ circuits. Furthermore, a
$\sum\prod^{[\alpha]}\sum\prod^{[\beta]}$ circuit is a $\sum\prod\sum\prod$
circuit such that the fan-in of the multiplication gates at level $3$
is bounded by $\alpha$, and the fan-in of the multiplication gates at
level $1$ is bounded by $\beta$. For example, a 
$\sum\prod^{[\alpha]}\sum\prod^{[\beta]}$ circuit computes
a polynomial of the form:
\begin{align*}
  \sum_{i=1}^t \prod_{j=1}^{a_i} \sum_{k=1}^{u_{i,j}} \prod_{l=1}^{b_{i,j,k}} x_{i,j,k,l}
\end{align*}
where $a_i\leq \alpha$, $b_{i,j,k}\leq \beta$.

Finally, in the following, we want to transform some circuits. The
underlying ring will be the same for the new circuit. Moreover, it can
be noticed that the following results (except for
Proposition~\ref{pro_1} and Corollary~\ref{cor_1}) hold for any
commutative ring.

\section{Upper bounds}

Here, we state the main theorem of this paper. 

\begin{theorem}\label{Thm_main}
  Let $f$ be an $n$-variate polynomial computed by a circuit of size
  $s$ and of degree $d$.  Then $f$ is computed by a
  $\sum\pn{O(\alpha)}\sum\pn{\beta}$ circuit $C$ of size
  $2^{O\left(\sqrt{d\log(ds)\log n}\right)}$ where
  $\alpha=\sqrt{d\frac{\log n}{\log ds}}$ and $\beta=\sqrt{d\frac{\log
      ds}{\log n}}$.  Furthermore, if $f$ is homogeneous, it will be
  also the case for $C$.
\end{theorem}

The previous theorem can be directly applied for the permanent.

\begin{theorem}
  If the $n\times n$ permanent is computed by a circuit of size
  polynomial in $n$, then it is also computed by a
  $\sum\pn{O(\sqrt{n})}\s\pn{O(\sqrt{n})}$ circuit of size
  $2^{O\left(\sqrt{n}\log(n)\right)}$.
\end{theorem}

In their paper~\cite{GKKS13}, Gupta, Kamath, Kayal and Saptharishi
used the previous $2^{\sqrt{d}\log^2(s)}$ bound~\cite{Koi12} for
parallelizing at depth $3$. They showed that:
\begin{proposition}[Theorem 1.1~in~\cite{GKKS13}]\label{pro_1}
  Let $f(x)\in\mathbb{Q}[x_1,\ldots,x_n]$ be an $n$-variate polynomial
  of degree $d = n^{O(1)}$ computed by an arithmetic circuit of size
  $s$. Then it can also be computed by a $\sum\prod\sum$ circuit of
  size $2^{O( \sqrt{d \log n \log s \log d})}$ with coefficients
  coming from $\mathbb{Q}$.
\end{proposition}

In fact, their proof is divided into three parts. First they transform
circuits into depth-$4$ circuits, then they transform depth-$4$
circuits into depth-$5$ circuits using only sum and exponentiation
gates. And finally they transform these last circuits into depth-$3$
circuits.  Using Theorem~\ref{Thm_main} instead of Theorem~4.1 in
their paper improves the first part of their proof. That implies a
small improvement of Theorem~1.1 in~\cite{GKKS13}:

\begin{corollary}\label{cor_1}
  Let $f(x)\in\mathbb{Q}[x_1,\ldots,x_n]$ be an $n$-variate polynomial
  of degree $d=n^{O(1)}$ computed by an arithmetic circuit of size
  $s$. Then it can also be computed by a $\sum\prod\sum$ circuit of
  size $2^{O(\sqrt{d\log n\log s})}$ with coefficients coming from
  $\mathbb{Q}$.
\end{corollary}

Finally, the use of the rationnals is important in the third part of
their proof. It will not be important in this paper.

\section{Useful propositions}

For proving Theorem~\ref{Thm_main}, we will need the following
propositions.

The next result is folklore. A proof can be found in~\cite{AJMV98}.

\begin{proposition}\label{Prop_homogeneous}
  If $f$ is a degree-$d$ polynomial computed by a
  $\{+,\times\}$-circuit $C$ of size $s$ such that the fan-in of each
  $+$-gate is unbounded and the fan-in of each $\times$-gate is
  bounded by $2$, then there exists a circuit $\tilde{C}$ of size $s
  (d+1)^2$ with $d+1$ outputs $O_0, O_1, \ldots, O_d$ such that:
  \begin{itemize}
  \item the fan-in of each $+$-gate is unbounded,
  \item the fan-in of each $\times$-gate is bounded by $2$,
  \item for each $i$, the gate $O_i$ computes the homogeneous part of
    $f$ of degree $i$,
  \item $\tilde{C}$ is homogeneous,
  \end{itemize} 
\end{proposition}

\begin{lemma}\label{lem_1} 
  In a homogeneous circuit, all the gates compute homogeneous
  polynomials. Moreover, the degree of each gate equals the
  degree of the homogeneous polynomial computed by this gate.
\end{lemma}

\begin{proof} 
  We show this lemma by induction on the underlying graph.
  \begin{itemize}
  \item The lemma is true for all the input gates.
  \item If $\alpha$ is a $+$-gate of inputs $\alpha_1,\ldots ,
    \alpha_p$, then by homogeneity, these inputs have the same degree
    $d$. By induction hypothesis, the gates
    $\alpha_1,\ldots,\alpha_p$ compute some homogeneous polynomials of
    degree $d$. So $[\alpha]$ is a homogeneous polynomial of degree
    $d$ or $-\infty$. By the remark~\ref{rem_zero}, the degree of
    $[\alpha]$ is $d$.
  \item If $\alpha$ is a $\times$-gate (or a $\odot$-gate) of inputs
    $\alpha_1,\ldots,\alpha_p$, then by induction hypothesis the
    polynomials $[\alpha_1],\ldots,[\alpha_p]$ are homogeneous and
    their degrees correspond to the degrees of
    $\alpha_1,\ldots,\alpha_p$. Hence $[\alpha]$ is homogeneous and
    the degree of $[\alpha]$ equals the degree of $\alpha$.
  \end{itemize}
\end{proof}

We define {\it $\times$-balanced} $\{\times,+,\odot\}$-circuits.
\begin{definition}
  A $\{\times,+,\odot\}$-circuit $C$ is called $\times$-balanced if
  and only if all the following properties are verified:
  \begin{itemize}
  \item the fan-in of each $\times$-gate is at most $5$,
  \item the fan-in of each $+$-gate is unbounded,
  \item the fan-in of each $\odot$-gate is at most $2$,
  \item for each $\times$-gate $\alpha$, each one of its arguments is
    of degree at most half of the degree of $\alpha$.
  \end{itemize}
\end{definition}
The last condition can not be true for the multiplication by a
scalar. It is the reason, we introduced the operator $\odot$.

The next proposition was found by Agrawal and Vinay~\cite{AV08}. It slightly
strengthens Valiant, Skyum, Berkowitz and Rackoff’s famous result~\cite{VSBR83}
by adding a constraint on all the $\times$-gates.

\begin{proposition}\label{Prop_logdepth}
  Let $f$ be a homogeneous degree-$d$ polynomial computed by a size-$s$ circuit
  $\tilde{C}$ verifying the four points of the conclusion of Proposition~\ref{Prop_homogeneous}.
  Then $f$ is computed by a homogeneous $\times$-balanced $\{\times,+,\odot\}$-circuit of size $s^6+s^4+1$ and of degree $d$. 
\end{proposition}

We present a proof of it in Section~\ref{Sec_5} as the statement above is
slightly different from the one we can find in~\cite{AV08} or in~\cite{SY10} (the
constants are a bit improved).

\begin{corollary}
  Let $f$ be a polynomial of degree $d$ computed by a circuit of size
  $s$. Then $f$ is computed by a $\{+,\times\}$-circuit of size $(sd)^{O(1)}$ and of depth
  $O(\log(s)\log(d))$ where each $+$ and $\times$-gate is of fan-in $2$.
\end{corollary}

\section{Proof of Proposition~\ref{Prop_logdepth}}\label{Sec_5}

Let $f$ be a homogeneous polynomial computed by a circuit $\tilde{C}$
of size $s$ such that
\begin{itemize}
\item the fan-in of each $+$-gate is unbounded,
\item the fan-in of each $\times$-gate is bounded by $2$,
\item $\tilde{C}$ is homogeneous.
\end{itemize}

First, we can assume that all the internal vertices are of positive
degree. To do that, we just have to replace recursively each gate
such that all entries are of degree $0$ by the constant value of this
gate. Then, by homogeneity, constants can not be entries of a
$+$-gate. Then, for each $\times$-gate such that one entry is a
constant, we replace the $\times$-gate by a scalar $\odot$-gate. We
can notice that this transformation does not increase the size of the
circuit.  Second, we can reorder the children of the $\times$-gates
and of the $\odot$-gates such that for each one of these gates, the
degree of the rightmost child is larger or equals the degree of the
other child. We get a circuit $C_1$ of size $s$.
  
We define now a new circuit $C_2$ which satisfies the criteria of the
proposition.  For each pair of gates $\alpha$ and $\beta$ in $C_1$, we
define the gate $(\alpha;\beta)$ in $C_2$ as follows (we will see
after how to compute it):
\begin{itemize}
\item If $\beta$ is a leaf, then $[(\alpha;\beta)]$ equals the sum of
  the parse trees rooted in $\alpha$ such that $\beta$ appears in the
  rightmost path (i.e., the leaf of the rightmost path corresponds to
  the gate $\beta$).
\item If $\beta$ is not a leaf, then $[(\alpha;\beta)]$ equals the sum
  of the parse trees rooted in $\alpha$ such that $\beta$ appears in
  the rightmost path and where the subtree rooted in this rightmost
  gate $\beta$ is deleted. That is as if we replace the rightmost
  appearance of the gate $\beta$ by the input $1$ and we compute
  $[(\alpha;\beta)]$ with $\beta = 1$ a leaf.
\end{itemize}
We notice here that it is easy to get the polynomial computed by the
gate $\alpha$:
\begin{align*}
  [\alpha] &= \s_{T_{\alpha}\textrm{ parse tree}} \textrm{value}(T_\alpha) \\
  & = \s_{l\textrm{ leaf of }C_1} \s_{\substack{T_\alpha\textrm{ parse
        tree s.t.} \\ \textrm{the rightmost leaf of }T_\alpha \\
      \textrm{is a copy of }l}} \textrm{value}(T_\alpha)\\
  &=\sum_{l\textrm{ leaf of }C_1}[(\alpha;l)].
\end{align*} We can notice that the number of parse trees can be
exponential but the last sum is of polynomial size.

Now, we show how one can compute the value of the gates
$(\alpha;\beta)$.
\begin{itemize}
\item If $\beta$ does not appear on the rightmost path of a parse tree
  rooted in $\alpha$, then $(\alpha;\beta)=0$.
\item In the case $\alpha=\beta$, if $\alpha$ is a leaf, then
  $(\alpha,\beta)=\alpha$ and else $(\alpha,\beta)=1$.
\item Otherwise $\alpha$ and $\beta$ are two different gates and
  $\alpha$ is not a leaf. If $\alpha$ is a $+$-gate, then
  $[(\alpha;\beta)]$ is simply the sum of all
  $[(\alpha^\prime,\beta)]$, where $\alpha^\prime$ is a child of
  $\alpha$.
\item If $\alpha$ is a $\odot$-gate, then one child is a constant $c$
  and the other child is a gate $\alpha^\prime$. Then $(\alpha;\beta)$
  is simply the scalar operation
  $[(\alpha;\beta)]=[(c;c)]\odot[(\alpha^\prime;\beta)]$.
\item If $\alpha$ is a $\times$-gate. There are two cases.
  \begin{itemize}
  \item First case: $\beta$ is a leaf. Then
    $\deg(\alpha)>\deg(\beta)$ and $\deg(\beta)\leq 1$. On each rightmost path ending on
    $\beta$ of a parse tree rooted in $\alpha$, there exists exactly
    one $\times$-gate $\gamma$ and its right child on this path
    $\gamma_r$ such that:
    \begin{align}\label{Eq_gamma1}
      \deg(\gamma) > \frac 1 2 \deg(\alpha) \geq \deg(\gamma_r).
    \end{align}
    Conversely, we notice that for each gate $\gamma$
    satisfying~(\ref{Eq_gamma1}), if $[(\alpha;\gamma)]$ and
    $[(\gamma_r;\beta)]$ are not zero, then $\gamma$ is on a rightmost
    path from $\alpha$ to $\beta$.  Then,
    \begin{align*}
      [(\alpha;\beta)] =\sum_{l\textrm{ leaf, }\gamma \
        \times\textrm{-gate verifying
        }(\textrm{\ref{Eq_gamma1}})}[(\alpha;\gamma)][(\gamma_l;l)][(\gamma_r;\beta)].
    \end{align*}

    As $\beta$ is a leaf, $\deg(\alpha;\beta)=\deg(\alpha)$.
    Using~(\ref{Eq_gamma1}):
    \begin{align*}
      \deg(\alpha;\gamma)=\deg(\alpha)-\deg(\gamma) & <
      \deg(\alpha)/2 \\
      \deg(\gamma_r;\beta)=\deg(\gamma_r) & \leq
      \deg(\alpha)/2 \\
      \deg(\gamma_l;l)=\deg(\gamma_l)\leq \deg(\gamma_r) & \leq
      \deg(\alpha)/2.
    \end{align*}
    Consequently, $[(\alpha;\beta)]$ is computed by a depth-$2$
    circuit of size at most $s^2+1$: a $+$-gate, of fan-in $s^2$, where each child is a
    $\times$-gate of fan-in $3$. Each child of these $\times$-gates is
    of degree at most the half of the degree of the $\times$-gate.
  \item Second case: $\beta$ is not a leaf. Then there exists on every
    rightmost paths rooted in $\alpha$ a $\times$-gate $\gamma$ and
    its child on this path $\gamma_r$ such that:
    \begin{align}\label{Eq_gamma2}
      \deg(\gamma) \geq (\deg(\alpha)+\deg(\beta))/2 > \deg(\gamma_r).
    \end{align}
    Then by the same argument,
    \begin{align}\label{Eq_multcase2}
      [(\alpha;\beta)] =\sum_{l\textrm{ leaf, }\gamma \
        \times\textrm{-gate verifying
        }(\textrm{\ref{Eq_gamma2}})}[(\alpha;\gamma)][(\gamma_l;l)][(\gamma_r;\beta)].
    \end{align}
    We have this time with~(\ref{Eq_gamma2}):
    \begin{align*}
      \deg(\alpha;\beta) & =\deg(\alpha)-\deg(\beta) \\
      \deg(\alpha;\gamma)=\deg(\alpha)-\deg(\gamma) & \leq
      \left(\deg(\alpha) -\deg(\beta) \right)/2\\
      \deg(\gamma_r;\beta)=\deg(\gamma_r) & < \left(\deg(\alpha)
        -\deg(\beta) \right)/2.
    \end{align*}
    The problem here is that the degree of $(\gamma_l;l)$ could be
    larger than $(\deg(\alpha)-\deg(\beta))/2$. The gate $\alpha$ is a
    $\times$-gate and its left child is of positive degree (otherwise
    $\alpha$ would be a $\odot$-gate). Hence,
    $\deg(\alpha;\beta)>\deg(\gamma_l;l)$. If $\gamma_l$ is of degree
    at most $1$ (and so exactly $1$ since $\gamma$ is not a
    $\odot$-gate), then $(\alpha;\beta)$ is of degree at least
    $2$. The computation of the gate $(\alpha;\beta)$ by the
    formula~(\ref{Eq_multcase2}) works (i.e., the degree of
    $(\gamma_l;l)$ is smaller than half of the degree of
    $(\alpha;\beta)$). Otherwise, the degree of $\gamma_l$ is at least
    $2$ and at most $\deg(\alpha;\beta)$. As $l$ is a leaf, we can
    apply the first case to the gate $\gamma_l$ (even if $\gamma_l$ is
    not a $\times$-gate). There exists also on every rightmost paths
    ending on $l$ and rooted in $\gamma_l$ a $\times$-gate $\mu$ and
    its child on this path $\mu_r$ such that:
    \begin{align}\label{Eq_mu}
      \deg(\mu) > \deg(\gamma_l)/2 \geq \deg(\mu_r).
    \end{align}
    Then,
    \begin{align*}
      [(\gamma_l;l)]=\s_{\substack{l_2\textrm{ leave of }C_1 \\
          \mu \ \times\textrm{-gate verifying
          }(\textrm{\ref{Eq_mu}})}}
      [(\gamma_l;\mu)][(\mu_l;l_2)][(\mu_r;l)].
    \end{align*}
    And so,
    \begin{align}
      [(\alpha;\beta)] =\sum_{ l,l_2,\gamma,\mu}
      [(\alpha;\gamma)][(\gamma_r;\beta)][(\gamma_l;\mu)][(\mu_l;l_2)][(\mu_r;l)].
    \end{align}
    where the sum is taken over all $l,l_2$ leaves of $C_1$, $\gamma$
    $\times$-gate verifying (\ref{Eq_gamma2}) and $\mu$ $\times$-gate
    verifying (\ref{Eq_mu}).
      
    The degrees of the gates $(\gamma_l;\mu)$, $(\mu_l;l_2)$ and
    $(\mu_r;l_1)$ are bounded by half of the degree of
    $\gamma_l$. Hence, $[(\alpha;\beta)]$ is computed by a depth-$2$
    size-$s^4+1$ circuit. The $\times$-gates are of fan-in bounded by
    $5$ and the degree of their children is bounded by half their
    degree.
  \end{itemize}
\end{itemize}
  
Consequently, for each gates $\alpha$ and $\beta$ in $C_1$, the gate
$(\alpha;\beta)$ is computed in $C_2$ by a sub-circuit of size at most
$s^4+1$. At the end we get a circuit of size at most $s^6+s^2$ which
computes all gates $(\alpha;\beta)$. Finally, $f$ is computed by a
circuit of size bounded by $s^6+s^2+1$.

That proves the proposition.

\section{Proof of Theorem~\ref{Thm_main}\label{Sec_proofdepth4}}

For realizing the reduction to depth four, Koiran begins by
transforming the circuit into an equivalent arithmetic branching
program. Then, he parallelizes the branching program, and finally
comes back to the circuits. The problem with this strategy is that the
transformation from circuits to branching programs requires an
increase in the size of our object. If the circuit is of size $s$, our
new branching program is of size $s^{\log(d)}$. Here, the approach is
to directly parallelize the circuit without using arithmetic branching
programs in intermediate steps.

The idea is to split the circuit into two parts: gates of degree lower
than $\sqrt{d}$ and gates of larger degree. Furthermore, a circuit
such that the degree of each gate is bounded by $\sqrt{d}$ computes a
degree-$\sqrt{d}$ polynomial and so can be written as a sum of at most
$s^{O(\sqrt{d})}$ monomials. Then, if each part of our circuit
computes polynomials of degrees bounded by $\sqrt{d}$, we just have to
get the two depth-$2$ circuits and connect them together. The main
difficulty comes from the fact it is not always true that the
sub-circuit obtained by the gates of degree larger than $\sqrt{d}$ is
of degree smaller than $\sqrt{d}$. For example, for the comb graph
with $n-1$ $\times$-gates and $n$ variable inputs:
\begin{align*}
  x_1\cdot
  \left(x_2\cdot\left(x_3\cdot\left(\ldots\right)\right)\right)
\end{align*}
the degree of the first part is $\sqrt{n}$, but the degree of the
second one is $n-\sqrt{n}$. In fact, we will show that this problem
does not happen if we just consider $\times$-balanced graphs. In this
case, the two parts have a degree bounded by $\sqrt{d}$.

Moreover, following ideas from~\cite{GKKS13}, we are going to cut not
exactly at level $\sqrt d$. It will give a sharper result.

\begin{lemma}\label{Lem_main}
  Let $f$ be a homogeneous $n$-variate polynomial of degree $d$
  computed by a homogeneous $\times$-balanced
  $\{\times,+,\odot\}$-circuit $C$ of size $\sigma$.  Then $f$ is
  computed by a homogeneous $\sum\prod^{[15a]}\sum\prod^{\left[\frac d
      a\right]}$ circuit of size $1+\binom{\sigma
    +15a}{15a}+\sigma+\sigma \binom{n+\frac d a}{\frac d a}+n$ for any
  positive constant $a$ smaller than $d$.
\end{lemma}

To get nicer expressions, we will use the following consequence of
Stirling's formula: (A proof appears in~\cite{AV08})

\begin{lemma}\label{Lem_stirling}
  \begin{align*}
    \binom{k+l}l = 2^{O\left(l+l\log \frac k l\right)}
  \end{align*}
\end{lemma}

First, let us see how Lemma~\ref{Lem_main} implies
Theorem~\ref{Thm_main}.

\begin{proof}[Proof of Theorem~\ref{Thm_main}]
  Let $f$ be an $n$-variate polynomial computing by a circuit of size
  $s$ and degree $d$. Let $\tilde{C}$ be the homogeneous circuit for
  the polynomial that we get by
  Proposition~\ref{Prop_homogeneous}. The circuit $\tilde{C}$ is of
  size $t=s(d+1)^2$ and computes all polynomials $f_0, \ldots, f_d$
  where $f_i$ is the homogeneous part of $f$ of degree $i$.  Then by
  Proposition~\ref{Prop_logdepth}, for each $i\leq d$, there exists a
  homogeneous $\times$-balanced circuit $C$ of size $\sigma=t^6+t^4+1$
  computing $f_i$.  We apply Lemma~\ref{Lem_main} for the circuit $C$
  with $a=\sqrt{d\frac{\log n}{\log \sigma}}$. Using
  Lemma~\ref{Lem_stirling} we get a homogeneous
  $\sum\pn{O(\alpha)}\sum\pn{\beta}$ circuit of size $1+\binom{\sigma
    +15a}{15a}+\sigma+\sigma \binom{n+\frac d a}{\frac d
    a}+n=2^{O\left(\sqrt{d\log \sigma \log n}\right)}$ with
  $\alpha=\sqrt{d\frac{\log n}{\log \sigma}}$ and $\beta=\sqrt{d\frac{\log
      \sigma}{\log n}}$.  At the end, we just have to add together
  homogeneous parts $f_i$.  As $\sigma=O(s^6d^{12})$, it gives a
  $2^{O\left(\sqrt{d\log(ds)\log n}\right)}$ upper bound for the size.
\end{proof}


Proving Lemma~\ref{Lem_main} will complete the proof.

\begin{proof}[Proof of Lemma~\ref{Lem_main}]
  We define $C_1$ and $C_2$ subcircuits of $C$ as follows. $C_1$ is
  the subcircuit of $C$ we get by keeping only gates of $C$ of degree
  $<\frac d a$. Circuit $C_2$ is made up of the remaining gates (i.e.,
  those of degree $\geq \frac d a$) and of the inputs of these
  gates. These inputs are the only gates which belong both in $C_1$
  and in $C_2$.

  Each gate $\alpha$ of $C_1$ has degree at most $\frac d a$, so
  computes a polynomial of degree at most $\frac d a$. By homogeneity
  of $C$, the polynomial computed in $\alpha$ is
  homogeneous. Consequently, $\alpha$ is a homogeneous sum of at most
  $\binom{n+\frac d a}{\frac d a}$ monomials, and so, can be computed
  by a homogeneous depth-2 circuit of size $1+\binom{n+\frac d
    a}{\frac d a}+n$ (The ``$1$'' encodes the $+$-gate, the ``$n$''
  encodes the input gates, and the remainder encodes the
  $\times$-gates).

  We are going to show now that the degree of $C_2$ is bounded by
  $15a$.

  Let $\delta$ be the degree of $C_2$. There exists a degree-$\delta$
  monomial $m$ in $C_2$. Let $T$ be a parse tree computing $m$.

  We can notice that a gate of $C_2$ can occur in many parse trees, and
  that in a parse tree one could find several copies of a gate of $C_2$.

  We partition the set of $\times$-gates of $T$ into $3$ sets:
  \begin{itemize}
  \item $\mathcal{G}_0 = \{\alpha \in T | \alpha \textrm{ is a
    }\times\textrm{-gate and all children of }\alpha\textrm{ are
      leaves of }T \}$
  \item $\mathcal{G}_1 = \{\alpha \in T | \alpha \textrm{ is a
    }\times\textrm{-gate and exactly one child of }\alpha\textrm{ is
      not a leaf} \}$
  \item $\mathcal{G}_2 = \{\alpha \in T | \alpha \textrm{ is a
    }\times\textrm{-gate and at least two children of }\alpha\textrm{
      are not leaves} \}$.
\end{itemize}

Then, if we consider the sub-tree $S$ of $T$ where the gates of $S$
are exactly the gates of $T$ which do not appear in $C_1$, then
$\mathcal{G}_0$ are leaves of $S$, $\mathcal{G}_1$ are internal
vertices of fan-in $1$ and $\mathcal{G}_2$ are internal vertices of
fan-in at least $2$.

The proof is in two parts. First we upperbound the size of the sets
$\mathcal{G}_0$, $\mathcal{G}_1$ and $\mathcal{G}_2$. Then, we
upperbound the degree of $m$.

In $C$, by Lemma~\ref{lem_1}, the degree of $m$ is at least the sum of
the degrees of the gates of $\mathcal{G}_0$ (since two of these gates
can not appear on the same path). Each one of these gates is in $C_2$,
so is of degree at least $\frac d a$ in $C$. As $m$ is of degree at
most $d$ in $C$, it means that the number of gates in $\mathcal{G}_0$
is at most $a$.

In $C$, alway by Lemma~\ref{lem_1}, the degree of $m$ is at least the
sum of the degrees of the leaves of $C_2$ directly connected to a gate
of $\mathcal{G}_1$. For each gate $\alpha$ of $\mathcal{G}_1$, exactly
one of its inputs $\beta$ is in $C_2$, hence of degree at least $\frac
d a$ in $C$.  By Proposition~\ref{Prop_logdepth}, the degree of
$\alpha$ is at least two times the degree of $\beta$, it yields that
the sum of degrees of inputs of $\alpha$ which are in $C_1$ is also at
least $\frac d a$. Then, the number of vertices in $\mathcal{G}_1$ is
at most $a$.

Finally, in a tree, the number of leaves is larger than the number of
vertices of fan-in at least $2$. Then in $S$, we get that:
\[
|\mathcal{G}_2|\leq |\mathcal{G}_0| \leq a.
\]

In $C_2$, the degree of the monomial $m$ is the number of leaves
labelled by a non-constant leaf in $T$. We match each leaf with the
first $\times$-gate which is connected to it. As in $T$, the fan-in of
the $\times$-gates is bounded by $5$, the fan-in of the $+$-gates is
bounded by $1$ and each $\odot$-gates add only one constant input,
then the number of variable leaves connected to a particular
$\times$-gate is at most $5$. So the number of leaves in $T$ is at
most:
\[
5\times \left(|\mathcal{G}_0| + |\mathcal{G}_1| +
  |\mathcal{G}_2|\right)\leq 15a.
\]

This proves that the degree of $C_2$ is at most $15a$. Then, the
number of inputs of $C_2$ is bounded by the number of gates in $C_1$
and so in $C$ (which is $\sigma$).  So, there exists a depth-$2$
circuit which compute $C_2$, of size $1+\binom{\sigma
  +15a}{15a}+\sigma$ with as inputs the gates of $C_1$.

Consequently, each polynomial $f_i$ can be computed by a homogeneous
$\sum\prod^{[a]}\sum\prod^{\left[\frac d a\right]}$ circuit of size at
most $1+\binom{\sigma +15a}{15a}+\sigma+\sigma \binom{n+\frac d
  a}{\frac d a}+n$.
\end{proof}

\section{A lower bound}\label{Sec_lowerbound}

In~\cite{GKKS12}, it was proved that if a homogeneous depth-four
circuit computing $\lang{Perm}_n$ has its bottom fan-in bounded by
$t$, then the size of the circuit is at least $2^{\Omega\left(\frac n
    t \right)}$. But what happens if bottom multiplication gates all
have a large fan-in? We show that this implies a similar lower bound
for the size of the circuit:

\begin{theorem}\label{Thm_lowbound>t}
  If $C$ is a homogeneous $\s\p\s\p$ circuit which computes
  $\lang{Perm}_n$ (or $\lang{Det}_n$) such that the fan-in of each
  bottom multiplication gate is at least $t$, then the size of $C$ is
  at least $2^{\Omega\left( t \log(n)\right)}$.
\end{theorem}

Our approach is only based on counting the number of monomials. We
begin by some definitions.
\begin{definition}
  For a multivariate polynomial $f(\mathbf{x})=\sum_{i=1}^{m_f} a_i
  \mathbf{x_i}$, we will denote $\mathcal{M}_f$ the set
  $\{\,\mathbf{x_i}\mid \mathbf{x_i} \textrm{ is a monomial of
  }f\,\}$. If $E$ is a set of polynomials, we also define
  $\mathcal{M}_E=\bigcup_{f\in E}\mathcal{M}_f$.
\end{definition}

We can notice $\mathcal{M}_{\lang{Perm}_n}=\{\,x_{1,\sigma(1)}\ldots
x_{n,\sigma(n)}\mid \sigma \in \mathfrak{S}_n\,\}$. So, $\left|
  \mathcal{M}_{\lang{Perm}_n} \right|=n!$.

\begin{definition}
  Let $E$ be a set of polynomials.  Let us denote
  \begin{align*}
    E^+=\{\,f_1+\ldots +f_m\mid m \in \mathbb{N} \textrm{ and }
    \forall
    i\leq m,\, f_i\in E\,\} \\
    \textrm{and }\ E^{\times k}=\{\,f_1\times \ldots \times f_m\mid m
    \leq k \textrm{ and } \forall i\leq m,\, f_i\in E\,\}
  \end{align*}
\end{definition}

\begin{lemma}\label{Lem_M+*}
  Let $E$ be a set of polynomials.  Then,
  \begin{align*}
    \mathcal{M}_{E^+} = \mathcal{M}_E \textrm { and } 
    \lvert \mathcal{M}_{E^{\times s}}\rvert \leq
  \left(\left|\mathcal{M}_E\right|+1\right)^s.
\end{align*}
\end{lemma}

\begin{proof}
  If $\mathbf{x}$ is a monomial in $\mathcal{M}_{E^+}$, it means there
  exist polynomials $f_1,\ldots,f_m$ in $E$ such that $\mathbf{x}$ is
  a monomial of $f_1+\ldots+f_m$. Then there exists $i\leq m$ such
  that $\mathbf{x}$ is a monomial of $f_i$ and so $\mathbf{x}$ is an
  element of $\mathcal{M}_E$. Hence $\mathcal{M}_{E^+} \subseteq
  \mathcal{M}_E$. Moreover, as $E\subseteq E^ +$, we get $
  \mathcal{M}_E\subseteq \mathcal{M}_{E^+} $.

  Moreover, if $\mathbf{x}$ is a monomial in $\mathcal{M}_{E^{\times
      s}}$, it means there exist polynomials $f_1,\ldots,f_m$ in $E$
  such that $\mathbf{x}$ is a monomial of $f_1\times \ldots \times
  f_m$ with $m\leq s$. It implies that $\mathbf{x} \in
  \{\,\mathbf{x_1}\times \ldots \times \mathbf{x_m}\mid m\leq s
  \textrm{ and } \mathbf{x_i}\in \mathcal{M}_E\,\}$. That is to say,
  $\mathbf{x} \in \{\,\mathbf{x_1}\times \ldots \times
  \mathbf{x_s}\mid \textrm{ and } \mathbf{x_i}\in
  \left(\mathcal{M}_E\cup \{1\}\right)\,\}$. It proves the lemma.
\end{proof}

Let $C$ be a $\sum\prod\sum\prod$ circuit. The gates of the circuit
are layered into five levels. Inputs are at level 0, multiplication
gates at levels 1 and 3 and addition gates at levels 2 and 4. For each
level $i$, let us denote $s_i$ the number of gates at this level,
$t_i$ an upper bound on the fan-in of these gates and $E_i$ the set of
polynomials computed at this level.

\begin{lemma}\label{Lem_lowbound<v}
  Any $\s\p\s\p$ circuit that computes $\lang{Perm}_n$ (or
  $\lang{Det}_n$) such that the fan-in of the multiplication gates at
  level 3 is bounded by $v$ must have size
  $\exp\left[\Omega\left(\frac n v \log(n) \right)\right]$.
\end{lemma}

\begin{proof}
  We notice that the hypothesis in the lemma about the bound of the
  fan-in just states that $t_3\leq v$.

  The polynomials in $E_1$ are just monomials. So,
  $\left|\mathcal{M}_{E_1}\right| \leq s_1$. We have:
  \begin{align*}
    E_4 \subseteq E_3^+, \ E_3 \subseteq E_2^{\times t_3} \textrm{ and
    }E_2 \subseteq E_1^+.
  \end{align*}
  Then by Lemma~\ref{Lem_M+*},
  \begin{align*}
    \left|\mathcal{M}_{E_4}\right| \leq (s_1+1)^{t_3} \leq (s_1+1)^v.
  \end{align*}
  However, as $\lang{Perm}_n$ is an element of $E_4$, we also have:
  \begin{align*}
    \left|\mathcal{M}_{E_4}\right| \geq
    \left|\mathcal{M}_{\lang{Perm}_n}\right| = n!.
  \end{align*}
  So, $s_1 \geq (n!)^{\frac 1 v}-1 = 2^{\Omega\left(\frac n v
      \log(n)\right)}$
\end{proof}

The result of this lemma directly implies
Theorem~\ref{Thm_lowbound>t}.
\begin{proof}[Proof of Theorem~\ref{Thm_lowbound>t}]
  Let $C$ be a homogeneous $\s\p\s\p$ circuit which computes
  $\lang{Perm}_n$ (or $\lang{Det}_n$) such that the fan-in of each bottom gate is at
    least $t$. It implies that the degree of each gate at level 1 and
    2 is at least $t$. As the circuit is homogeneous, the degree of
    a gate at level 3 is upperbounded by $n$ and lowerbounded by
    $t$ times the number of inputs of this gate. Consequently, in $C$,
    the fan-in of the multiplication gates at level 3 is bounded by
    $\frac n t$. Then Lemma~\ref{Lem_lowbound<v} implies the theorem.
\end{proof}

In fact, for computing the determinant, we can also notice that the fan-in of multiplication gates in
the depth-four circuits that we get either in~\cite{Koi12} or here in
Section~\ref{Sec_proofdepth4}, is linear in $\sqrt{n}$. It implies that in this case, the bounds are tight.

\begin{corollary}
  If $C$ is a $\s\p\s\p$ circuit which computes $\lang{Det}_n$ such that the fan-in of
  each bottom multiplication gate is $\Omega(\sqrt{n})$ or such that
  the fan-in of each multiplication gate of level $3$ is $O(\sqrt{n})$, 
  then the minimal size of $C$ is $2^{\Theta\left(\sqrt{n} \log(n)\right)}$.
\end{corollary}

\begin{proof}
  Koiran's result~\cite{Koi12} implies that there exist depth-four circuits for
  $\lang{Det}_n$ of size $2^{O(\sqrt{n}\log n)}$ such that all
  multiplication gates have fan-in bounded by $O(\sqrt{n})$.
  For the lowerbound, the case where the bottom fan-in is lowerbounded by
  $\Omega(\sqrt{n})$ is given by Theorem~\ref{Thm_lowbound>t}. The case where the fan-in of
  gates of level $3$ is bounded by $O(\sqrt{n})$ is given by Lemma~\ref{Lem_lowbound<v}.
\end{proof}

It would be interesting to know the lower bound on
the size of an homogeneous circuit computing
$\lang{Det}_n$. In~\cite{GKKS12} the authors show that if the circuit
is such that the fan-in of bottom gates is bounded by $O(\sqrt{n})$,
then the size is $2^{\sqrt{n}}$. Here, we show that if all bottom
fan-in are lowerbounded by $\Omega(\sqrt{n})$, then the size is
$2^{\Omega(\sqrt{n}\log n)}$. What happens if in the circuit, there are some
bottom gates with a large fan-in and some bottom gates with a small fan-in?

\begin{question}
  Is it true that if $\mathcal{C}$ is a homogeneous depth-four circuit
  which computes $\lang{Det}_n$ then the size of $\mathcal{C}$ is at
  least $2^{\Omega(\sqrt{n})}$? 
\end{question}

{\small
\section*{Acknowledgments} 

The author thanks Pascal Koiran for helpful discussions and
comments on this work.}

\bibliographystyle{plain}
\bibliography{parallelisation}

 \end{document}